\documentclass[11pt,a4paper]{amsart}

\usepackage[latin1]{inputenc}
\usepackage[english]{babel}
\usepackage{amsmath}

\usepackage{amssymb}
\usepackage[numbers]{natbib}
\usepackage{graphicx}
\usepackage{braket}
\usepackage[pdftex,plainpages=false,colorlinks,hyperindex,bookmarksopen,linkcolor=red,citecolor=blue,urlcolor=blue]{hyperref}

\usepackage{epstopdf}

\usepackage{braket}

\bibpunct{[}{]}{;}{n}{,}{,}
\usepackage[hmargin=3cm,vmargin={3.5cm,4cm}]{geometry}

\theoremstyle{definition}     
\newtheorem{thm}{Theorem}                
\newtheorem{prop}{Proposition}      
\newtheorem{lemma}{Lemma}      

\theoremstyle{definition}                           

\theoremstyle{remark}                             
\newtheorem*{rmk}{Remark}              

\usepackage{color}

\newcommand{\be}{\begin{eqnarray}}
\newcommand{\ee}{\end{eqnarray}}

\newcommand{\R}{\mathbb{R}}  
\newcommand{\C}{\mathbb{C}} 
\newcommand{\N}{\mathbb{N}} 

\newcommand{\wt}[1]{\widetilde{#1}}

\def\J{{\widetilde{J}}}

\newcommand{\ber}[2]{\text{ber}_{#1}\left(#2\right)}
\newcommand{\bei}[2]{\text{bei}_{#1}\left(#2\right)}

\numberwithin{equation}{section}

\allowdisplaybreaks

\begin{document}
\title{Energy dissipation in viscoelastic Bessel media}

    \author{Ivano Colombaro$^1$}
		\address{${}^1$ Department of Information and Communication Technologies, Universitat 
		Pompeu Fabra, C/Roc Boronat 138, Barcelona, Spain.}
		\email{ivano.colombaro@upf.edu}
	
	    \author{Andrea Giusti$^2$}
		\address{${}^2$ Institute for Theoretical Physics, ETH Zurich, Wolfgang-Pauli-Strasse 27,
		8093, Zurich, Switzerland.}	
 		\email{agiusti@phys.ethz.ch}
	
\author{Andrea Mentrelli$^3$}
\address{${}^3$ Department of Mathematics $\&$ Alma Mater Research Center on Applied
		Mathematics (AM$^2$), University of Bologna,
		Via Saragozza 8, Bologna, Italy.}	
 		\email{andrea.mentrelli@unibo.it}

    \date  {\today}

\begin{abstract}
We investigate the specific attenuation factor for the Bessel models of viscoelasticity. We find that the quality factor for this class can be expressed in terms of Kelvin functions and that its asymptotic behaviours confirm the analytical results found in previous studies for the rheological properties of these models.
\end{abstract}

    \maketitle
    
\section{Introduction}

Viscoelasticity represents one of the most compelling and vibrant research topics in continuum mechanics, both from the perspectives of applied sciences and mathematics. For comprehensive reviews of the history of this field and its modern developments we refer the interested readers to, {\em e.g.}, \cite{pipkin2012lectures,RogosinMainardi2014,mainardi2011creep,Mainardi2010book}. In this context, it is of particular interest the role played by non-local operators, with particular regard for fractional ones \cite{mainardi2011creep,Mainardi2010book,giusti2020practicalguide}. More precisely, fractional derivatives are, loosely speaking, mathematical objects belonging to a subclass of {\em weakly-singular Volterra-type convolution integro-differential operators}, for further details see, {\em e.g.}, \cite{MainardiGorenflo1997FC,Samko1993fractint,diethelm2020why}.

In this work we present a study of the processes of storage and dissipation of energy for a specific class of models of linear viscoelasticity, known as {\em Bessel models} \cite{colombaro2017bessel}. To this end, we shall analytically compute the so-called {\em quality factor}, {\em i.e.} $Q$-factor, \cite{Mainardi2010book,borcherdt2009viscoelastic} starting from the Laplace representation of the creep compliance for a viscoelastic medium governed by the Bessel constitutive laws \cite{colombaro2017bessel}. To carry out our analysis we will mostly follow \cite[Sect. 2]{colombaro2018storage}, that consists of a coherent summary of the arguments in \cite{Mainardi2010book,borcherdt2009viscoelastic}.

The work is therefore organised as follows. In Section~\ref{sec:BesselMod} we review the {\em creep representation} for the Bessel models of linear viscoelasticity and their generalities. In Section~\ref{sec:Qfactor} we derive explicit expressions for the $Q$-factor for these models in terms of special functions of particular interest. Section~\ref{sec:numerical} presents some numerical results and illuminating plots for the computed $Q$-factor. Lastly, in Section~\ref{sec:conclusion} we summarise the main results of the study and provide some concluding remarks. 

\subsection*{Acknowledgments}
The authors would like to thank Francesco Mainardi for helpful discussions.
A.G. is supported by the European Union's Horizon 2020 research and innovation programme under the Marie Sk\l{}odowska-Curie Actions (grant agreement No.~895648). 
A.M. is partially supported by the PRIN2017 project ``Multiscale phenomena in Continuum Mechanics: singular limits, off-equilibrium and transitions'' (Project Number: 2017YBKNCE).
The work of the authors has also been carried out in the framework of the activities of the Italian National Group of Mathematical Physics [Gruppo Nazionale per la Fisica Matematica (GNFM), Istituto 
Nazionale di Alta Matematica (INdAM)].
\section{Bessel Models in Linear Viscoelasticity} \label{sec:BesselMod}
	In linear viscoelasticity the constitutive relation for a uniaxial homogeneous and isotropic viscoelastic body in the {\em creep representation} reads \cite{Mainardi2010book}
\be
\label{eq:constitutive}
\varepsilon (t) = J_{\rm g} \, \sigma (t) + \int_0 ^t \dot{J} (t-t') \, \sigma (t) \, {\rm d} t' \, ,
\ee
where $\varepsilon (t)$ and $\sigma (t)$ are respectively the (uniaxial) stress and strain functions, $J(t)$ is the material function known as {\em creep compliance}, $J_{\rm g} := J(0^+) \geq 0$ is the glass compliance, and the dot denotes a derivative with respect to time. Note that $J(t)$ is a {\em causal function} and hence it vanishes for $t<0$. Additionally, we define the {\em rate of creep (compliance)} for a viscoelastic system as
\be
\label{eq:memory}
\Psi (t) := \frac{\dot{J}(t)}{J_{\rm g}} \, ,
\ee
that keeps track of the memory effects in the model.

Under some loose regularity conditions one can Laplace transform both sides of Eqs.~\eqref{eq:constitutive} and \eqref{eq:memory}, that yields 
\be
\wt{\varepsilon} (s) = s \wt{J} (s) \wt{\sigma} (s)	\quad \mbox{and} \quad 
s\wt{J} (s) = 1 + \wt{\Psi} (s) \, ,
\ee
with $s \in \C$ the complex Laplace frequency and 
\be
\mathcal{L}\left\{ f(t) \, ; \, s \right\} \equiv \wt{f} (s) := \int_0 ^\infty \mathrm{e}^{-st} \, f(t) \, {\rm d} t \, ,
\ee
denoting the Laplace transform of a sufficiently regular causal function $f(t)$.

The {\em Bessel models} \cite{colombaro2017bessel} are a class of viscoelastic models characterised by a creep rate $\Psi (t;\nu)$ expressed, for $\nu >-1$, in terms of the Dirichlet series
\be
\label{eq:dir}
\Psi (t;\nu) = 4(\nu+1)(\nu+2) +  4(\nu+1)\sum_{k=1}^\infty \exp\left( -j_{\nu+2,k}^2 t \right)	\,,
\ee
where $j_{\nu+2,k}$ are the $k$-th positive real root of the Bessel function of the first kind $J_{\nu + 2}(t)$ (see, {\em e.g.}, \cite{Abramowitz1965handbook} for details). Note that this series is absolutely convergent for $t>0$. 

Taking the Laplace transform of Eq.~\eqref{eq:dir} one finds
\be 
\label{eq:Psi-laplace}
\wt{\Psi} (s;\nu) :=  \frac{2 \, (\nu + 1)}{\sqrt{s}} \, \frac{I _{\nu + 1} (\sqrt{s})}{I _{\nu + 2} (\sqrt{s})} \, ,
\ee
where $I_{\alpha} (z)$ denotes the modified Bessel functions of the first kind \cite{Abramowitz1965handbook}
\be
I_{\alpha} (z) := \left( \frac{z}{2} \right) ^{\alpha} \sum _{m=0} ^\infty \frac{1}{m! \, \Gamma (m+\alpha+1)} \left( \frac{z}{2} \right) ^{2m} \,,
\ee
with $\Gamma (z)$ representing the Euler Gamma function. Then, as showed in~\cite{colombaro2017bessel}, one finds that the creep compliance for the Bessel models reads
\begin{equation}
\label{eq:Bessel-sJ(s)}
s \J (s;\nu) = 1 + \widetilde{\Psi} (s; \nu) =
1+ \frac{2 (\nu + 1)}{\sqrt{s}} \frac{I_{\nu + 1} (\sqrt{s})}{I_{\nu + 2} (\sqrt{s})} \, , 
\end{equation}
that, taking advantage of the identity \cite{Abramowitz1965handbook}
$$
I_{\nu-1}\left(z\right) - I_{\nu+1}\left(z\right) = \frac{2\nu}{z} I_{\nu}\left(z\right) \, ,
$$
can be recast as 
\begin{equation}
	\label{eq:Bessel-sJ(s)-bis}
	s \J (s;\nu) = \frac{I_{\nu} (\sqrt{s})}{I_{\nu + 2} (\sqrt{s})} \, .
\end{equation}
This expression will be the starting point for computing the $Q$-factor for the Bessel models.

Before moving on to the explicit computation of the quality factor for these models it is worth taking some time to highlight the origin and main results concerning this class of viscoelastic systems. To start off, it is worth mentioning that the Bessel viscoelastic class was formulated in \cite{colombaro2017bessel} as a generalisation of a mathematical model for the propagation of blood pulses within large arteries \cite{giusti2016dynamic}. The mathematical techniques developed in \cite{colombaro2017bessel,giusti2016dynamic} were the used to provide an alternative derivation of the Rayleigh-Sneddon sum \cite{giusti2016infiniteseries}. In \cite{giusti2017infiniteorder} it was shown that the constitutive relations of the Bessel models are {\em ordinary infinite-order} differential equations, whereas the short time behaviour for these systems effectively reduces to a fractional Maxwell model of order $1/2$ and $\nu$-dependent relaxation time. Lastly, taking advantage of the Buchen-Mainardi algorithm \cite{buchen1975asymptotic} (see \cite{colombaro2017transientwaves} for a review on the subject) the propagation of transient waves in a semi-infinite Bessel medium was investigated, deriving the precise form of the wave-front expansion.

\begin{rmk}
The time variable $t$, in this section and in the following, is effectively non-dimensional since, for the sake of convenience, we have set the relaxation time $\tau$ to unity. 
\end{rmk}
\section{Quality factor for the Bessel models} \label{sec:Qfactor}
The specific attenuation factor or quality factor, often abbreviated as $Q$-factor, is a non-dimensional quantity that measures the dissipation of energy for sinusoidal excitations in stress or strain \cite{Mainardi2010book,borcherdt2009viscoelastic,colombaro2018storage,mainardi2019genbecker}.

Following \cite{Mainardi2010book} and \cite{colombaro2018storage}, given a complex creep compliance $\wt{J}(s)$, in the Laplace domain, one can obtain the corresponding $Q$-factor as \cite{Mainardi2010book}
\be
\label{eq:Q-general}
Q^{-1} (\omega) = -\frac{\Im\left\{ s \, \J (s) \vert_{s=\imath\omega} \right\} }{\Re \left\{ s \, \J (s) \vert_{s=\imath\omega} \right\} }	\,,
\ee
where $\omega>0$ is the frequency of the harmonic excitations of the material.

First, let us define {\em Tricomi's uniform modified Bessel functions of the first kind} (in analogy with Tricomi's uniform Bessel functions discussed in \cite{Mainardi2010book}) as follows:
\be
\label{eq:tricomif}
I^{\rm T}_\alpha (z) := \left( \frac{z}{2} \right) ^{-\alpha} \, I_\alpha (z) = \sum _{m=0} ^\infty \frac{1}{m! \, \Gamma (m+\alpha+1)} \left( \frac{z}{2} \right) ^{2m} \, .
\ee
For these functions one has that:
\begin{lemma} \label{lemma-1}
Let $\alpha \in \R$, $\alpha>-1$, and $z \in \C$. Then, $I^{\rm T}_\alpha (z)$ is an entire function and $I^{\rm T}_\alpha (\sqrt{z})$ is both single-valued and entire.
\end{lemma}
\begin{proof}
First, if $z=0$ it is clear from Eq.~\eqref{eq:tricomif} that $I^{\rm T}_\alpha (0) = 1/\Gamma (\alpha+1)$ which is surely well defined for $\alpha>-1$. If instead $z \in \C \setminus \{0\}$ and
$$ a_k (z) = \frac{1}{k! \, \Gamma (k+\alpha+1)} \left( \frac{z}{2} \right) ^{2k} \, , $$
then $\left| a_{k+1} (z)/a_k (z) \right| = O (k^{-2})$ and $\left| a_{k+1} (\sqrt{z})/a_k (\sqrt{z}) \right| = O (k^{-2})$ as $k \to +\infty$. Hence the series in $I^{\rm T}_\alpha ({z})$ and $I^{\rm T}_\alpha (\sqrt{z})$ converge everywhere in $\C$ as a consequence of the ratio test. 
Second, from the definition \eqref{eq:tricomif} one has that
$$
I^{\rm T}_\alpha (\sqrt{z}) = \sum _{m=0} ^\infty \frac{1}{m! \, \Gamma (m+\alpha+1)} \left( \frac{z}{4} \right) ^{m} \, ,
$$
which is clearly a single-valued function on $\C$.
\end{proof}

Therefore, it follows that:
\begin{prop}\label{proposition-1}
$s \wt{J}(s)$ as in Eq.~\eqref{eq:Bessel-sJ(s)-bis} is single-valued and
\be
\label{eq-intermediate}
s \J (s;\nu) = \frac{4}{s} \frac{I^{\rm T}_{\nu} (\sqrt{s})}{I^{\rm T}_{\nu + 2} (\sqrt{s})} \, .
\ee
\end{prop}

\begin{proof}
From Eqs.~\eqref{eq:Bessel-sJ(s)-bis} and \eqref{eq:tricomif} one has that
\be
s \J (s;\nu) = \frac{I_{\nu} (\sqrt{s})}{I_{\nu + 2} (\sqrt{s})} = \frac{4}{s} \frac{I^{\rm T}_{\nu} (\sqrt{s})}{I^{\rm T}_{\nu + 2} (\sqrt{s})} \, ,
\ee 
which is therefore single-valued as a consequence of Lemma~\ref{lemma-1}.
\end{proof}

The last proposition allows one to compute $s \J (s;\nu) \big\vert _{s=\imath\omega}$ without the risk of incurring in branch cuts and branch points for positive real frequencies $\omega$.

Let us continue with some preliminary results.

\begin{lemma} \label{lemma-2}
Let $\omega \in \R$, $\omega > 0$, and $m \in \N$. Then,
\be
(\imath \omega)^m = \left\{
\begin{aligned}
& (-1)^n \omega^{2n} \, , \,\, \qquad \mbox{if} \,\, m=2n \, , \,\, n \in \N \,\,\, \mbox{(even)} \\
& \imath (-1)^n \omega^{2n+1} \, , \quad \mbox{if} \,\, m=2n+1 \, , \,\, n \in \N \,\,\, \mbox{(odd)}
\end{aligned}
\right.
\ee
\end{lemma}
\begin{proof}
Trivial.
\end{proof}

Now, defining the two functions
\be
\label{eq:f}
f_\alpha (z) := \sum _{n=0} ^\infty \frac{(-1)^n z^{2n}}{2^{4n}(2n)! \, \Gamma (2n+\alpha +1)} \,,
\ee
\be
\label{eq:g}
g_\alpha (z) := \sum_{n=0}^\infty \frac{(-1)^n z^{2n+1}}{2^{4n+2}(2n+1)! \, \Gamma (2n+\alpha +2) }	\,,
\ee
One can conclude that:
\begin{lemma} \label{lemma-3}
Let $\omega, \alpha \in \R$, $\omega > 0$ and $\alpha > -1$. 
Then $f_\alpha (z)$ and $g_\alpha (z)$ are entire functions.
\end{lemma}
\begin{proof}
Again, for $z=0$ one gets that both $f_\alpha (0)$ and $g_\alpha (0)$ are finite. Furthermore, employing again the ratio test it is easy to see that the series in Eqs.~\eqref{eq:f} and \eqref{eq:g} converge for all $z \in \C \setminus \{0\}$.
\end{proof}

\begin{prop} \label{prop-2}
Let $\omega, \nu \in \R$, $\omega > 0$ and $\nu > -1$, then
\be
\label{eq:split}
I^{\rm T}_\nu (\sqrt{\imath \omega}) = f_\nu (\omega) + \imath \, g_\nu (\omega) \, ,
\ee
and
\be
s \J (s;\nu) \Big\vert _{s=\imath\omega} = \frac{4}{\imath\omega} \frac{f_\nu (\omega)  + \imath g_\nu(\omega)}{f_{\nu+2} (\omega)  + \imath g_{\nu+2} (\omega)}
\ee
\end{prop}
\begin{proof}
Consider the series representation \eqref{eq:tricomif}, {\em i.e.},
$$
I^{\rm T}_\nu (\sqrt{\imath \omega}) = \sum _{m=0} ^\infty \frac{(\imath \omega)^{m}}{2 ^{2m} \,m! \, \Gamma (m+\nu+1)} \, ,
$$
then taking advantage of Lemma~\ref{lemma-2} and \ref{lemma-3} one can split the right-hand side of this last expression and recognise that Eq.~\eqref{eq:split} holds. For the second part of the proof it suffices to replace Eq.~\eqref{eq:split} into Eq.~\eqref{eq-intermediate} in Proposition~\ref{proposition-1}.
\end{proof}

\begin{thm} \label{thm-1}
Let $\omega, \nu \in \R$, $\nu>-1$, and $\omega > 0$. Then the $Q$-factor for the Bessel models reads:
\be
\label{eq:q-bessel-1}
Q^{-1} \left(\omega; \nu\right) =
\frac{f_{\nu} (\omega)f_{\nu+2} (\omega) + g_{\nu} (\omega)g_{\nu+2} (\omega)}{g_{\nu} (\omega)f_{\nu+2} (\omega) - f_{\nu} (\omega)g_{\nu+2} (\omega)} \, .
\ee
\end{thm}
\begin{proof}
From Proposition~\ref{prop-2} it is easy to see that
\be
\nonumber
s \J (s;\nu) \Big\vert _{s=\imath\omega} &=& \frac{4}{\imath\omega} \frac{f_\nu (\omega)  + \imath g_\nu(\omega)}{f_{\nu+2} (\omega)  + \imath g_{\nu+2} (\omega)}\\
 &=& \nonumber \frac{4}{\omega} \frac{g_{\nu} (\omega)f_{\nu+2} (\omega) - f_{\nu} (\omega)g_{\nu+2} (\omega) - \imath \bigl(f_{\nu} (\omega)f_{\nu+2} (\omega) + g_{\nu} (\omega)g_{\nu+2} (\omega)\bigr)}{\vert f_{\nu+2} (\omega) \vert^2 + \vert g_{\nu+2} (\omega) \vert^2} \, ,
\ee
via a direct computation. This allows one to separate directly the real and imaginary parts of $s \J (s;\nu) \big\vert _{s=\imath\omega}$. Indeed, one finds
\be
\nonumber
\Re \left\{ s \, \J (s) \big\vert_{s=i\omega} \right\} &=&\frac{4}{\omega}
\frac{g_{\nu} (\omega)f_{\nu+2} (\omega) - f_{\nu} (\omega)g_{\nu+2} (\omega) }{\vert f_{\nu+2} (\omega) \vert^2 + \vert g_{\nu+2} (\omega) \vert^2} \, , \\
\nonumber
\Im \left\{ s \, \J (s) \big\vert_{s=i\omega} \right\}  &=& - \frac{4}{\omega}
\frac{f_{\nu} (\omega)f_{\nu+2} (\omega) + g_{\nu} (\omega)g_{\nu+2} (\omega)}{\vert f_{\nu+2} (\omega) \vert^2 + \vert g_{\nu+2} (\omega) \vert^2} \, .
\ee
Inserting these expressions into Eq.~\eqref{eq:Q-general} yields the final result.
\end{proof}

It is now interesting to introduce another couple of special functions. Specifically, let us consider the {\em Kelvin functions} \cite{Abramowitz1965handbook} $\ber{\alpha}{x}$ and $\bei{\alpha}{x}$. These functions are, respectively, the real and imaginary parts of $J_{\alpha}\left(x e^{i\frac{3}{4}\pi}\right)$, {\em i.e.},
\be
\label{eq:ber}
\ber{\alpha}{z} := \left(\frac{z}{2}\right)^\alpha  \sum_{k=0}^\infty
\frac{\cos\left[\left( \frac{3\alpha}{4}+\frac{k}{2}\right)\pi\right]}{k!\Gamma(k+\alpha+1)} \left(\frac{z^2}{4}\right)^k \, ,
\ee
\be
\label{eq:bei}
\bei{\alpha}{z} := \left(\frac{z}{2}\right)^\alpha  \sum_{m=0}^\infty
\frac{\sin\left[\left( \frac{3\alpha}{4}+\frac{m}{2}\right)\pi\right]}{m!\Gamma(m+\alpha+1)} \left(\frac{z^2}{4}\right)^m \, .
\ee

\begin{lemma} \label{lemma-k1}
Let $\alpha \in \R$ and $k \in \mathbb{Z}$. Then,
\be
\nonumber
\cos\left( \frac{3\pi}{4}\alpha+\frac{k\pi}{2}\right)= 
\left\{
\begin{aligned}
& (-1)^n\cos\left(3\pi\alpha/4\right) \, , \qquad \mbox{if} \,\, m=2n \, , \,\, n \in \N \,\,\, \mbox{(even)} \\
& (-1)^{n+1}\sin\left(3\pi \alpha/4\right) \, , \quad \mbox{if} \,\, m=2n+1 \, , \,\, n \in \N \,\,\, \mbox{(odd)}
\end{aligned}
\right.
\ee
\be
\nonumber
\sin\left( \frac{3\pi}{4}\alpha+\frac{k\pi}{2}\right)= 
\left\{
\begin{aligned}
& (-1)^n\sin\left(3\pi\alpha/4\right) \, , \qquad \mbox{if} \,\, m=2n \, , \,\, n \in \N \,\,\, \mbox{(even)} \\
& (-1)^{n}\cos\left(3\pi \alpha/4\right) \, , \qquad \mbox{if} \,\, m=2n+1 \, , \,\, n \in \N \,\,\, \mbox{(odd)}
\end{aligned}
\right.
\ee
\end{lemma}
\begin{proof}
Trivial.
\end{proof}

\begin{prop} \label{prop-k1}
Let $\alpha,\omega \in \R$, $\alpha >-1$, and $\omega>0$. Then one finds
\be
\left(\frac{2}{\sqrt{\omega}}\right)^\alpha \ber{\alpha}{\sqrt{\omega}} =
\cos\left(  \frac{3\pi}{4}\alpha\right) f_\alpha (\omega) - \sin\left(  \frac{3\pi}{4}\alpha\right) g_\alpha(\omega) \, ,
\ee
\be
 \left(\frac{2}{\sqrt{\omega}}\right)^\alpha \bei{\alpha}{\sqrt{\omega}} =
 \sin\left(  \frac{3\pi}{4}\alpha\right) f_\alpha (\omega) +\cos\left(  \frac{3\pi}{4}\alpha\right) g_\alpha(\omega) \,.
\ee
\end{prop}
\begin{proof}
It follows from Lemma~\ref{lemma-k1} together with Eqs. \eqref{eq:f}, \eqref{eq:f}, \eqref{eq:ber}, and \eqref{eq:bei}.
\end{proof}

Now, it is fairly easy to see that the result in Proposition~\ref{prop-k1} can be rewritten as
\be
\label{eq:trasf-1}
f_\alpha (\omega) &=& \left(\frac{2}{\sqrt{\omega}}\right)^\alpha \left[ \cos\left(  \frac{3\pi}{4}\alpha\right) \ber{\alpha}{\sqrt{\omega}} +  \sin\left(  \frac{3\pi}{4}\alpha\right) \bei{\alpha}{\sqrt{\omega}}\right] \, ,
\\
\label{eq:trasf-2} g_\alpha (\omega) &=& \left(\frac{2}{\sqrt{\omega}}\right)^\alpha \left[ -\sin\left(  \frac{3\pi}{4}\alpha\right) \ber{\alpha}{\sqrt{\omega}} + \cos\left(  \frac{3\pi}{4}\alpha\right) \bei{\alpha}{\sqrt{\omega}}\right] \, ,
\ee
which means that one can recast the result in Theorem~\ref{thm-1} as follows.
\begin{thm} \label{thm-2}
Let $\omega, \nu \in \R$, $\nu>-1$, and $\omega > 0$. Then the $Q$-factor for the Bessel models reads:
\be
\label{eq:q-bessel-2}
Q^{-1}\left(\omega; \nu\right) = \frac{ \bei{\nu+2}{\sqrt{\omega}} \ber{\nu}{\sqrt{\omega}} - \bei{\nu}{\sqrt{\omega}} \ber{\nu+2}{\sqrt{\omega}} }{ \bei{\nu}{\sqrt{\omega}} \bei{\nu+2}{\sqrt{\omega}} + \ber{\nu}{\sqrt{\omega}} \ber{\nu+2}{\sqrt{\omega}} } \, .
\ee
\end{thm}

\begin{proof}
See Appendix~\ref{appendix-proof-thm2}.
\end{proof}

\begin{rmk}
Expressing $Q^{-1}\left(\omega; \nu\right)$ in terms of Kelvin functions turns out to be particularly useful for its numerical evaluation since these functions are already implemented in most of scientific softwares of common use.
\end{rmk}

\subsection{$Q$-factor for the fractional Maxwell model and Bessel media}
Let $\tau$ be a time scale, then the constitutive equation of the fractional Maxwell model of linear viscoelasticity reads \cite{Mainardi2010book}
\be
\sigma (t) + \tau^\beta \, {^{\rm C} \! D}^\beta \sigma (t) 
= J_{\rm g} ^{-1} \tau^\beta \, {^{\rm C} \! D}^\beta \epsilon (t)
\, , \quad 0 < \beta < 1 \, ,
\ee
where ${^{\rm C} \! D}^\beta$ denotes the Caputo derivative of order $\beta$ with respect to $t$. Note that the case $\beta = 1$ corresponds to the (ordinary) Maxwell model \cite{Mainardi2010book}. Following a procedure akin to the one discussed in the present section one can easily derive the specific dissipation function for this model (see \cite{Mainardi2010book} for details) which yields
\be
Q^{-1}_{\beta}(\omega) = \frac{\sin(\pi\beta/2)}{(\omega\tau)^\beta + \cos(\pi\beta/2)} \, ,
\ee 
again, with $0 < \beta < 1$. 

As shown in \cite{colombaro2017bessel}, the Bessel models approach the behaviour of a fractional Maxwell
model of order $1/2$ for short times ($t \to 0^+$) and of an ordinary Maxwell model for long times ($t \to +\infty$). More precisely, we have the following results for the creep compliance of the Bessel models.

\begin{lemma}[see \cite{colombaro2017bessel}] \label{lemmaccio}
Consider the creep compliance for the Bessel models in the Laplace domain, {\em i.e.} Eq.~\eqref{eq:Bessel-sJ(s)}. Then one finds
\be
s \J (s;\nu) \sim
\left\{
\begin{aligned}\label{eq:Qasympt}
& 1 + 2(\nu + 1) s^{-1/2} \, , \qquad \qquad \quad \mbox{as} \,\, s \to \infty \, , \\
& \frac{2 (\nu+2)}{\nu+3} + \frac{4 (\nu+1) (\nu+2)}{s} \, , \quad \mbox{as} \,\, s \to 0 \, ,
\end{aligned}
\right.
\ee 
with $\nu > -1$.
\end{lemma}

Then, one can easily infer the following proposition.
\begin{prop}\label{prop-asymp}
Let $\nu > -1$. The asymptotic behaviour of the $Q$-factor of the Bessel models is given by
\be
Q^{-1}\left(\omega; \nu\right) \sim
\left\{
\begin{aligned}\label{eq:Qasympt2}
& \frac{\sqrt{2}(\nu +1)}{\omega^{1/2} + \sqrt{2}(\nu+1)} \, , \quad \quad \mbox{as} \,\, \omega \to +\infty \, , \\
& \frac{2 (\nu+1) (\nu+3)}{\omega} \, , \qquad \quad \mbox{as} \,\, \omega \to 0^+ \, .
\end{aligned}
\right.
\ee 
\end{prop}
\begin{proof}
Set $s=\imath\omega$, with $\omega \in \R$ and $\omega>0$, in the results from Lemma~\ref{lemmaccio}. Note that the asymptotic expansion for $s \to \infty$ in Lemma~\ref{lemmaccio} looks multivalued (although the full function is single-valued), hence to perform the analysis one can simply choose the principal branch of $\sqrt{s}$.
\end{proof}

This result clearly shows that the high-frequency limit of $Q^{-1}\left(\omega; \nu\right)$ behaves as a fractional Maxwell model of order $1/2$ whereas, in a similar fashion, the low-frequency behaviour of $Q^{-1}\left(\omega; \nu\right)$ approaches the one of a standard Maxwell body.


\section{Numerical Results} \label{sec:numerical}

We shall now provide some numerical results and plots to elucidate the behaviour of $Q$-factor for the Bessel models governed by the analytical expression in Eq.~\eqref{eq:q-bessel-2}. Specifically, we will provide the numerical evaluation of Eq.~\eqref{eq:q-bessel-2}, against the frequency $\omega$, as for different values of the parameter $\nu>-1$.  
\begin{figure}[h!]
\centering
 \includegraphics[scale=0.8]{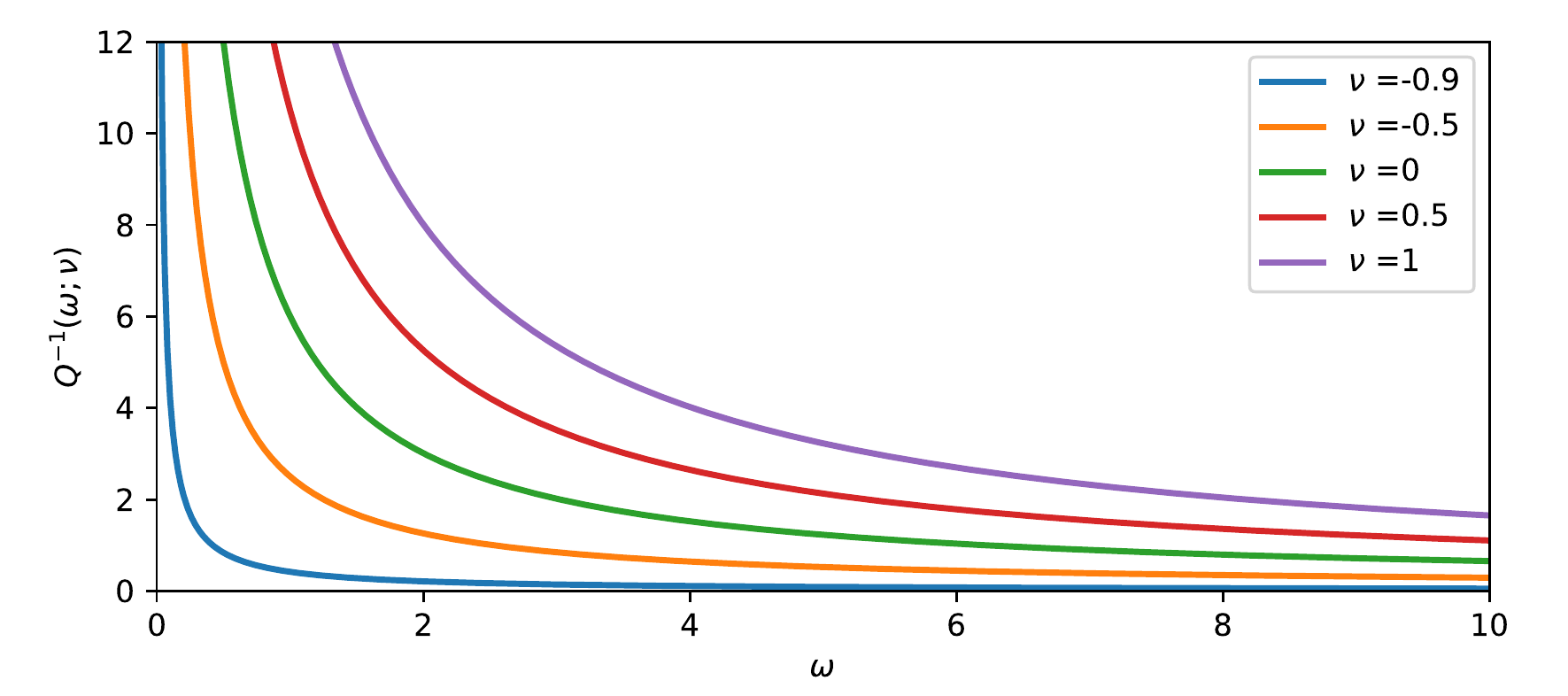} 
\caption{$Q^{-1}(\omega ; \nu)$ for different values of $\nu$, in linear scale.}
\label{fig:Q1linlin}
\end{figure}

\begin{figure}[h!]
\centering
 \includegraphics[scale=0.8]{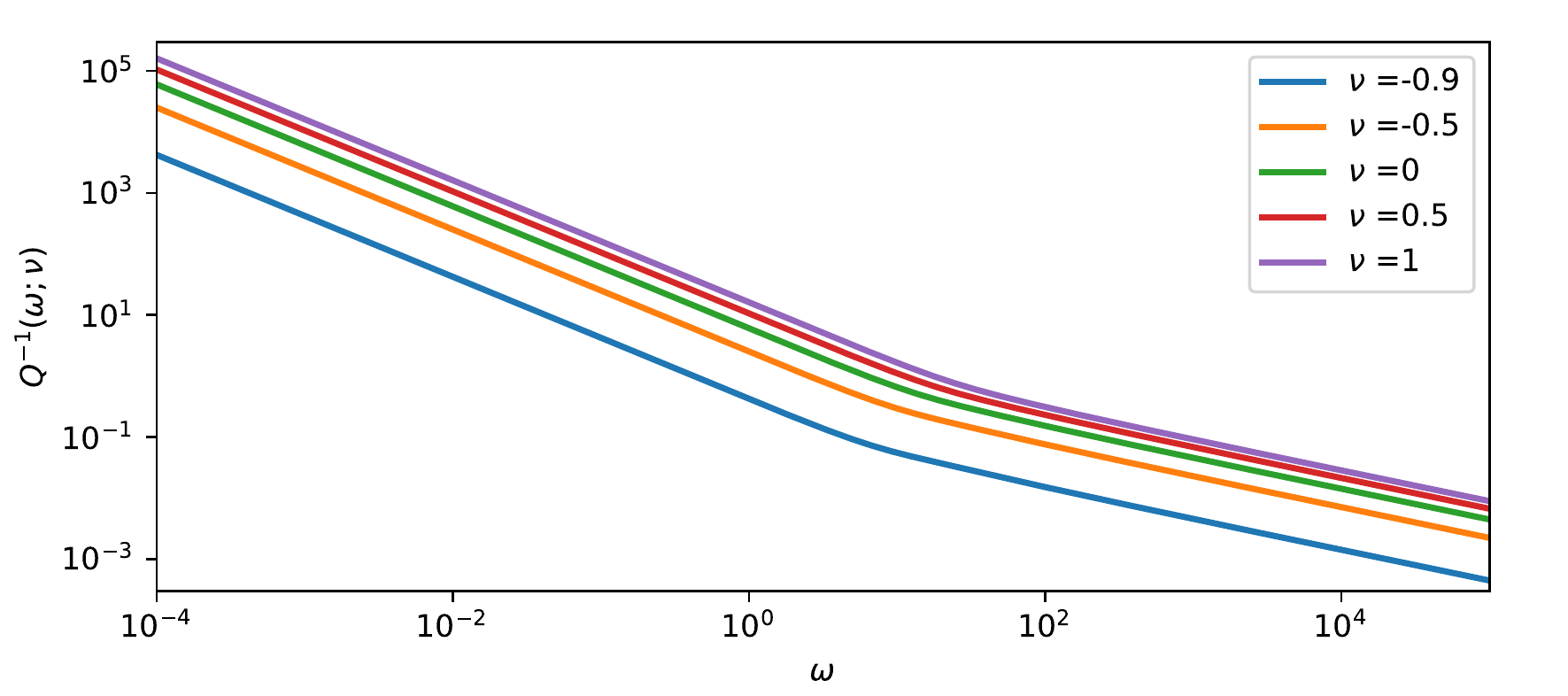} 
\caption{$Q^{-1}(\omega ; \nu)$ for different values of $\nu$, in logarithmic scale.}
\label{fig:Q1loglog}
\end{figure} 
In~\figurename~\ref{fig:Q1linlin} and~\figurename~\ref{fig:Q1loglog} we provide the plots of the quality factors, for different values of $\nu>-1$, in both linear and in logarithmic scales. From \figurename~\ref{fig:Q1linlin} one can appreciate that $Q^{-1}(\omega; \nu)$ is overall a decreasing function, with a steep behaviour at low frequencies and a softer one at high frequencies. The plot in logarithmic scale, \figurename~\ref{fig:Q1loglog}, shows the behaviour of the quality factor for an interval of values of the frequency larger than that of \figurename~\ref{fig:Q1linlin}, ranging from $10^{-4}$ to $10^5$. From \figurename~\ref{fig:Q1loglog} one can immediately identify two regions where $Q^{-1}(\omega; \nu)$ presents nearly constant slopes (in Log-Log scale), while the transition from one region to the other appears to be sharper for lower values of the parameter $\nu$.

\begin{figure}[h!]
\centering
\includegraphics[scale=0.8]{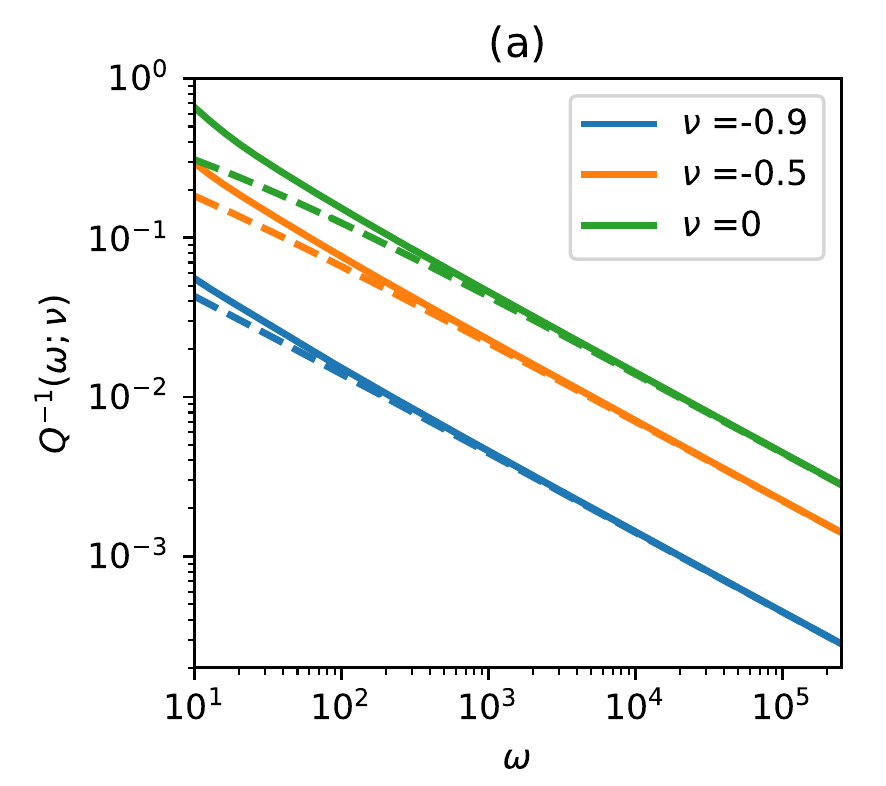}%
\includegraphics[scale=0.8]{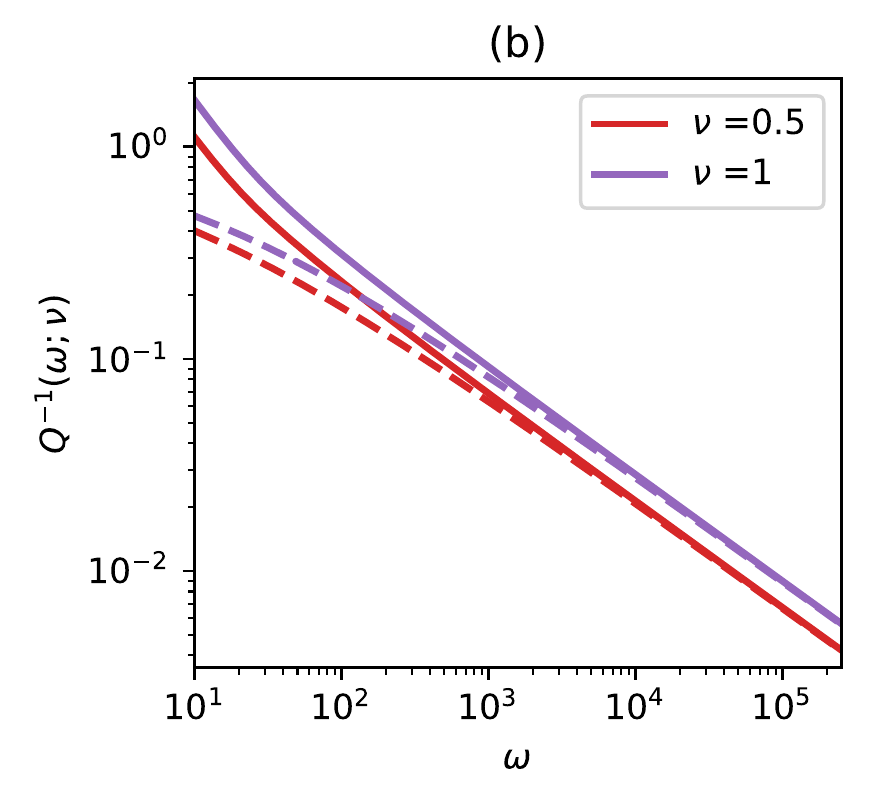} 
\caption{Comparison between $Q^{-1}(\omega ; \nu)$ (continuous line) and its asymptotic behavior (dashed line) for $\omega  \to \infty$, in logarithmic scale.}
\label{fig:Q1asinfty}
\end{figure}
\begin{figure}[h!]
\centering
\includegraphics[scale=0.8]{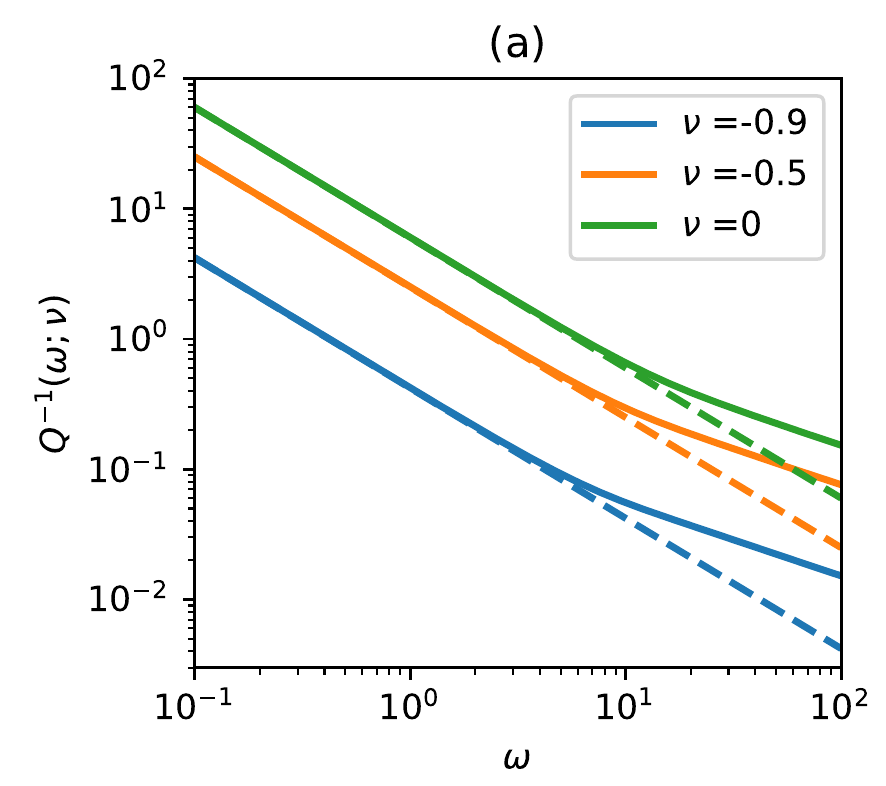}%
\includegraphics[scale=0.8]{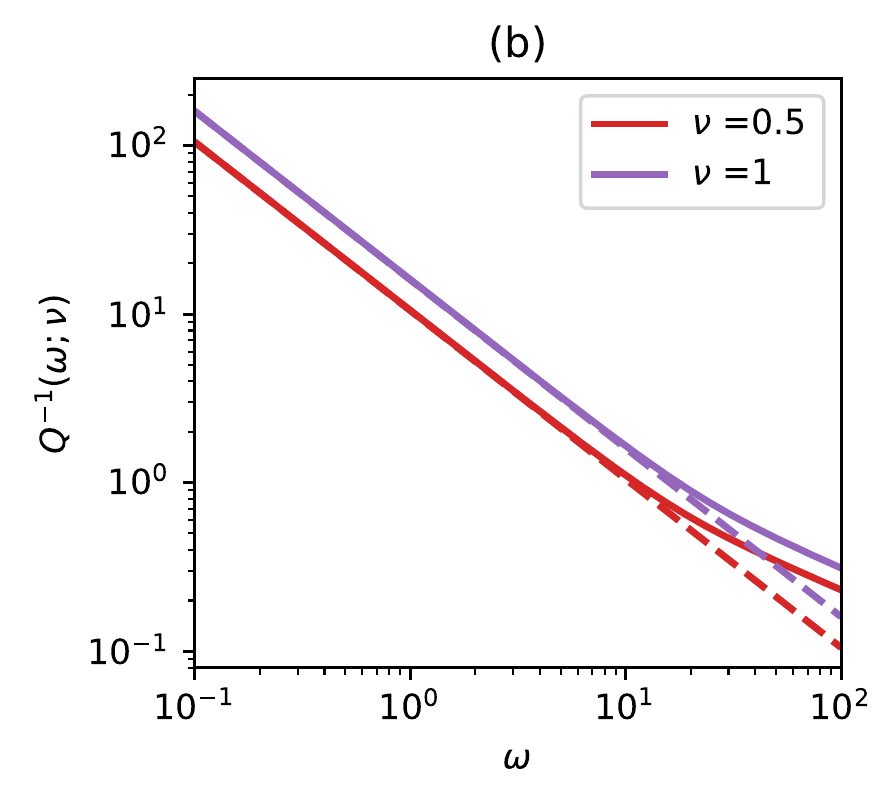} 
\caption{Comparison between $Q^{-1}(\omega ; \nu)$ (continuous line) and its asymptotic behavior (dashed line) for $\omega \to 0$, in logarithmic scale.}
\label{fig:Q1aszero}
\end{figure}

In \figurename~\ref{fig:Q1asinfty} we show the numerical matching between the analytic expression of the full $Q^{-1}(\omega ; \nu)$, Eq.~\eqref{eq:q-bessel-2}, and its estimated asymptotic behaviour at high frequencies ($\omega  \to \infty$) provided in Eq.~\eqref{eq:Qasympt2}$_1$. Similarly, in~\figurename~\ref{fig:Q1aszero} we show matching between Eq.~\eqref{eq:q-bessel-2} and its low-frequency asymptotic expansion provided in Eq.~\eqref{eq:Qasympt2}$_2$. These plots further highlight the fact that at short times (high frequencies) $Q$-factor of the Bessel models approaches the one of a fractional Maxwell model of order $1/2$, whereas at late times (low frequencies) the model relaxes to a standard Maxwell model, in accordance with the results in \cite{colombaro2017bessel} concerning the material and memory functions.

\section{Discussion and conclusions} \label{sec:conclusion} 
The Bessel models are a class of models of linear viscoelasticity that was originally derived in the context of hemodynamics \cite{giusti2016dynamic}. The constitutive laws for these models are infinite-order ordinary differential equations \cite{giusti2017infiniteorder} leading to material functions that, in the time domain, are expressed in terms of Dirichlet series, whereas in the Laplace domain they given by suitable ratios of modified Bessel functions of contiguous order \cite{colombaro2017bessel}.  

The specific attenuation factor, or $Q$-factor, is an important quantity in viscoelasticity that provides a quantitative estimate of the dissipation of energy for sinusoidal excitations in stress or strain due to the properties of the material \cite{Mainardi2010book,borcherdt2009viscoelastic,mainardi2019genbecker}.

In this work we have provided the full analytic derivation of the $Q$-factor for the Bessel models. Specifically, Theorem~\ref{thm-2} provides a precise expression for the $Q$-factor of these models in terms of a rate of Kelvin functions. Furthermore, in Proposition~\ref{prop-asymp} we provided a precise characterisation of the asymptotic behaviour of the $Q$-factor at both low frequencies and high frequencies. This asymptotic analysis agrees with previous findings concerning the matching between this class of models and the fractional Maxwell model of order $1/2$, at short times, and the standard Maxwell model, at long times \cite{colombaro2017bessel}. In other words, these models feature a continuous transition from a fractional-like behaviour to an ordinary one. Additionally, in Section~\ref{sec:numerical} we provided some numerical evaluations of the quantities computed in Section~\ref{sec:Qfactor} in order to elucidate on their full behaviour, that might not be apparent from the analytical results.

\appendix

\section{Proof of Theorem~\ref{thm-2}} \label{appendix-proof-thm2}
From Eqs. \eqref{eq:trasf-1} and \eqref{eq:trasf-2} one finds:
\begin{align*}
(i) \,\, f_{\nu} (\omega)f_{\nu+2} (\omega) = \left(\frac{2}{\sqrt{\omega}} \right)^{2\nu+2} 
\Biggl[
&\cos\left( \frac{3\pi}{4}\nu\right) \cos\left( \frac{3\pi}{4}(\nu+2)\right) \ber{\nu}{\sqrt{\omega}} \ber{\nu+2}{\sqrt{\omega}}
\nonumber \\
+& \cos\left( \frac{3\pi}{4}\nu\right) \sin\left( \frac{3\pi}{4}(\nu+2)\right) \ber{\nu}{\sqrt{\omega}} \bei{\nu+2}{\sqrt{\omega}}
\nonumber \\
+& \sin\left( \frac{3\pi}{4}\nu\right) \cos\left( \frac{3\pi}{4}(\nu+2)\right) \bei{\nu}{\sqrt{\omega}} \ber{\nu+2}{\sqrt{\omega}}
\nonumber \\
+& \sin\left( \frac{3\pi}{4}\nu\right) \sin\left( \frac{3\pi}{4}(\nu+2)\right) \bei{\nu}{\sqrt{\omega}} \bei{\nu+2}{\sqrt{\omega}}
\Biggr] \, ;
\end{align*}

\begin{align*} 
(ii) \,\, g_{\nu} (\omega)g_{\nu+2} (\omega) = \left(\frac{2}{\sqrt{\omega}} \right)^{2\nu+2} 
\Biggl[
&\sin\left( \frac{3\pi}{4}\nu\right) \sin\left( \frac{3\pi}{4}(\nu+2)\right) \ber{\nu}{\sqrt{\omega}} \ber{\nu+2}{\sqrt{\omega}}
\nonumber \\
-& \sin\left( \frac{3\pi}{4}\nu\right) \cos\left( \frac{3\pi}{4}(\nu+2)\right) \ber{\nu}{\sqrt{\omega}} \bei{\nu+2}{\sqrt{\omega}}
\nonumber \\
-& \cos\left( \frac{3\pi}{4}\nu\right) \sin\left( \frac{3\pi}{4}(\nu+2)\right) \bei{\nu}{\sqrt{\omega}} \ber{\nu+2}{\sqrt{\omega}}
\nonumber \\
+& \cos\left( \frac{3\pi}{4}\nu\right) \cos\left( \frac{3\pi}{4}(\nu+2)\right) \bei{\nu}{\sqrt{\omega}} \bei{\nu+2}{\sqrt{\omega}}
\Biggr] \, ;
\end{align*} 

\begin{align*}
(iii) \,\, g_{\nu} (\omega)f_{\nu+2} (\omega) = \left(\frac{2}{\sqrt{\omega}} \right)^{2\nu+2} 
\Biggl[
&-\sin\left( \frac{3\pi}{4}\nu\right) \cos\left( \frac{3\pi}{4}(\nu+2)\right) \ber{\nu}{\sqrt{\omega}} \ber{\nu+2}{\sqrt{\omega}}
\nonumber \\
-& \sin\left( \frac{3\pi}{4}\nu\right) \sin\left( \frac{3\pi}{4}(\nu+2)\right) \ber{\nu}{\sqrt{\omega}} \bei{\nu+2}{\sqrt{\omega}}
\nonumber \\
+& \cos\left( \frac{3\pi}{4}\nu\right) \cos\left( \frac{3\pi}{4}(\nu+2)\right) \bei{\nu}{\sqrt{\omega}} \ber{\nu+2}{\sqrt{\omega}}
\nonumber \\
+& \cos\left( \frac{3\pi}{4}\nu\right) \sin\left( \frac{3\pi}{4}(\nu+2)\right) \bei{\nu}{\sqrt{\omega}} \bei{\nu+2}{\sqrt{\omega}}
\Biggr] \, ;
\end{align*}

\begin{align*}
(iv) \,\, f_{\nu} (\omega)g_{\nu+2} (\omega) = \left(\frac{2}{\sqrt{\omega}} \right)^{2\nu+2} 
\Biggl[
-&\cos\left( \frac{3\pi}{4}\nu\right) \sin\left( \frac{3\pi}{4}(\nu+2)\right) \ber{\nu}{\sqrt{\omega}} \ber{\nu+2}{\sqrt{\omega}}
\nonumber \\
+& \cos\left( \frac{3\pi}{4}\nu\right) \cos\left( \frac{3\pi}{4}(\nu+2)\right) \ber{\nu}{\sqrt{\omega}} \bei{\nu+2}{\sqrt{\omega}}
\nonumber \\
-& \sin\left( \frac{3\pi}{4}\nu\right) \sin\left( \frac{3\pi}{4}(\nu+2)\right) \bei{\nu}{\sqrt{\omega}} \ber{\nu+2}{\sqrt{\omega}}
\nonumber \\
+& \sin\left( \frac{3\pi}{4}\nu\right) \cos\left( \frac{3\pi}{4}(\nu+2)\right) \bei{\nu}{\sqrt{\omega}} \bei{\nu+2}{\sqrt{\omega}}
\Biggr] \, .
\end{align*}

\noindent Then, plugging (i)-(iv) into Eq.~\eqref{eq:q-bessel-1} in Theorem~\ref{thm-1} one immediately infers Eq.~\eqref{eq:q-bessel-2}.
%
%
%

%
%
%
\end{document}